\documentclass[conference]{IEEEconf}

\IEEEoverridecommandlockouts                              % This command is only needed if 
\usepackage{bbm}
\usepackage[english]{babel}
\usepackage{amsmath,amsfonts,amssymb, amsthm}
\usepackage{graphicx}
\usepackage[colorlinks=true, allcolors=blue]{hyperref}
\usepackage{parskip}
\usepackage{tabularx}
\usepackage{bm}
\usepackage{hhline}
\usepackage{multirow}
\usepackage{algorithm}
\usepackage{algpseudocode}
\usepackage{etoolbox}
\usepackage{subcaption}
\usepackage{chapterbib}
\usepackage{cite}

\patchcmd{\thebibliography} %% removes spacing between bibliography items
  {\settowidth}
  {\setlength{\parsep}{0pt}\setlength{\itemsep}{0pt plus 0.1pt}\settowidth}
  {}{}

\newtheorem{lemma}{Lemma}

\newtheorem{proposition}{Proposition}

\newtheorem*{remark}{Remark}

\newcommand{\KLdiv}{\mathrm{D}_{\mathrm{KL}}}
\renewcommand{\t}{^{\mbox{\tiny\sf T}}} %% the transpose operator
\newcommand{\blkdiag}{\mathrm{blkdiag}}
\newcommand{\R}{\mathbb{R}}

\newcommand{\E}{\mathbb{E}}
\newcommand{\tr}{\mathrm{tr}}
\renewcommand{\det}{\mathrm{det}}

\newcommand{\cov}{\mathrm{cov}}
\newcommand{\half}{\mbox{$\frac{1}{2}$}}

\allowdisplaybreaks

\title{\LARGE \bf  Discrete-Time Maximum Likelihood Neural Distribution Steering}
\author{George Rapakoulias$^{1}$ and Panagiotis Tsiotras$^{2}$
\thanks{$^{1}$ Ph.D. student, School of Aerospace Engineering, Georgia Institute of Technology, Atlanta, GA, 30332, USA. Email:
 {\tt\small grap@gatech.edu}}%
\thanks{$^{2}$ David and Andrew Lewis Chair and Professor, School of Aerospace Engineering, and Institute for Robotics and Intelligent Machines, Georgia Institute of Technology, Atlanta, GA, 30332, USA. Email:
 {\tt\small tsiotras@gatech.edu}}%
}

\begin{document}
\maketitle

\begin{abstract}
This paper studies the problem of steering the distribution of a discrete-time dynamical system from an initial distribution to a target distribution in finite time.
The formulation is fully nonlinear, allowing the use of general control policies, parametrized by neural networks.
Although similar solutions have been explored in the continuous-time context, extending these techniques to systems with discrete dynamics is not trivial.
The proposed algorithm results in a regularized maximum likelihood optimization problem, which is solved using machine learning techniques. 
After presenting the algorithm, we provide several numerical examples that illustrate the capabilities of the proposed method. 
We start from a simple problem that admits a solution through semidefinite programming, serving as a benchmark for the proposed approach.
Then, we employ the framework in more general problems that cannot be solved using existing techniques, such as problems with non-Gaussian boundary distributions and non-linear dynamics.
\end{abstract}

\section{Intoduction}

The problem of controlling distributions of dynamical systems has attracted increased attention in the past few years, mainly due to its practical applications in machine learning and generative AI \cite{ruthotto2021introduction}.
Compared to standard optimal control problems, where a feedback policy that minimizes a functional in the presence of uncertainty is sought, distribution control aims at directly steering the distribution of the system's state while minimizing a cost function \cite{caluya2021wasserstein, chen2021controlling}.
Apart from generative AI applications, this approach is suitable for controlling systems whose behavior is better captured by a distribution rather than a deterministic state variable.
Such applications include, for example, swarm and multiagent control \cite{saravanos2023distributed, Saravanos-RSS-21}, mean field games \cite{ruthotto2020machine, chen2023density}, opinion dynamics \cite{liu2022deep}, and safe stochastic model predictive control \cite{knaup2023safe, saravanos2022distributed}, among many others.

This work focuses on the problem of controlling the distribution of a deterministic, discrete-time, control-affine system with nonlinear drift.  
The problem of interest can be cast as the following infinite-dimensional optimization problem
\begin{subequations} \label{reference_problem}
\begin{align}
& \min_{ x_k, u_k } \; J =  \sum_{k = 0}^{N-1} { \E \left[ \| u_k \|^2 + V_k(x_k)\right]}, \label{ref:cost} \\
& x_{k+1} = f_k(x_k) + B_k u_k, \label{ref:dyn} \\
& x_0 \sim \rho_i, \label{ref:initial_distr} \\
& x_N \sim \rho_f, \label{ref:final_distr} 
% & \P(x_k \in \mathcal{X}) \geq 1-\epsilon, \label{ref:chan_con_x}
\end{align}
\end{subequations}
where $\rho_i, \; \rho_f$ are the boundary state distributions, $f_k, B_k$ describe the system's prior dynamics, and $V_k(x_k)$ is a state-dependent cost that penalizes potentially undesirable regions of the state space.
Problem \eqref{reference_problem} is an Optimal Mass Transport problem with (nonlinear) prior dynamics (OMTwpd).
For continuous linear prior dynamics and $V(x_k) \equiv 0$, this problem is equivalent to an Optimal Transport (OT) problem in a transformed set of coordinates defined by the linear dynamics \cite{chen2016optimal}. 
Furthermore, for a quadratic state cost $V(x_k) = x_k\t Q_k x_k$, linear prior dynamics, and for Gaussian initial and terminal distributions, globally optimal solutions exist through the Covariance Steering (CS) framework for both continuous and discrete time settings. 
These can be computed efficiently using semidefinite programming (SDP) \cite{bakolas2018finite, liu2022optimal, balci2022exact, chen2015optimal, chen2015optimal2, rapakoulias2023discrete}.
For general, nonlinear systems or systems with non-Gaussian boundary distributions, a globally optimal solution is difficult to obtain. 
Local solutions through linearization have been explored in \cite{ridderhof2020chance}, while solutions utilizing characteristic functions, for non-Gaussian boundary distributions and noise, have been proposed in \cite{sivaramakrishnan2022distribution}. 
However, both of these methods utilize linear feedback policies and therefore, may be suboptimal.
Furthermore, the authors of \cite{balci2023density} propose a randomized policy for steering between Gaussian Mixture Models (GMM) with deterministic linear prior dynamics using randomized linear policies.
While the algorithm results in numerically efficient solutions with guaranteed convergence leveraging analytical results from the Covariance Steering theory, the optimality of the policy compared to more general nonlinear policies is not explored in \cite{balci2023density}. 

Focusing on formulations with nonlinear dynamics, most existing works concern systems in continuous time. 
In the simplest case where $\dot{x}_t = u_t$, the problem has been studied in the context of Continuous Normalizing Flows (CNFs) \cite{chen2018neural, onken2021ot} due to its applications in generative AI \cite{ruthotto2021introduction}.  
Its stochastic counterpart, referred to as the Schr\"{o}dinger Bridge Problem (SPB), has also been explored in a similar context \cite{chen2016relation}.
Works that account for more complicated prior dynamics usually impose some assumptions on their structure. 
For example, in \cite{liu2022deep, chen2022likelihood} the authors account for dynamics with a nonlinear drift term but otherwise require control in all states by restricting their analysis to cases where $B_k = I$, and focus on mean field games and generative AI applications.
A more general framework, allowing dynamics with fewer control channels than states is considered in \cite{caluya2020finite} where the authors extended the results of \cite{chen2016optimal} to feedback linearizable systems. 
Finally, a framework that does not require feedback linearizable dynamics is explored in \cite{caluya2021wasserstein} but the deterministic drift term is required to be the gradient of a potential function. 

The discrete-time problem has been studied much less.
The case of degenerate prior dynamics, i.e., $x_{k+1} = x_{k} + u_{k}$ has been addressed within the framework of Discrete Normalizing Flows (DNFs) \cite{behrmann2019invertible, huang2020convex, berg2018sylvester}. 
To the best of the authors' knowledge, the more general problem with nonlinear prior dynamics has not been addressed in the literature.
This problem is of interest for several reasons.
First, it captures the case where the system dynamics are inherently discrete, as in the digital implementation of control algorithms.
Furthermore, most types of neural networks can be analyzed through the lens of discrete dynamical systems, \cite{haber2017stable, he2016deep, pilipovsky2023probabilistic}, providing a significant drive towards more research in the discrete setting. 
Finally, from a computational point of view, solving the discrete-time steering problem requires storing the state vector at a constant number of intervals, requiring, therefore, a fixed memory budget, compared to continuous formulations which perform a temporal discretization based on the stiffness of the trained dynamics \cite{behrmann2019invertible}.  

To this end, in this work, we study Problem \eqref{reference_problem} using tools from machine learning and the DNF literature, specifically combining 
control-theoretic ideas and tools from \cite{behrmann2019invertible}. 
To bring Problem \eqref{reference_problem} to the framework of Normalizing Flows, we first relax the terminal distribution constraint \eqref{ref:final_distr} to a Kullback–Leibler (KL) divergence soft constraint, giving rise to the problem
\begin{subequations} \label{reference_problem_soft}
\begin{align}
& \hspace*{-3mm} \min_{ x_k, u_k } \; J = \sum_{k = 0}^{N-1} { \E \left[ \| u_k \|^2 + V_k(x_k) \right]} + \lambda \KLdiv (\rho_N \| \rho_f), \label{ref_soft:cost}\\
& x_{k+1} = f_k(x_k) + B_k u_k, \\
& x_0 \sim \rho_i,
% & \P(x_k \in \mathcal{X}) \geq 1-\epsilon, 
\end{align}
\end{subequations}
where $\lambda > 0$. 
Henceforth, this is the general problem formulation we will use in this paper.
%Problem \eqref{reference_problem_soft} is the general formulation we solve in this paper.

\section{Notation}
We use lowercase letters to denote vectors and vector random variables and capital letters to denote matrices. 
Given a function $F: \R^{n} \rightarrow \R^m$, its Jacobian is denoted by $\nabla F: \R^n \rightarrow \R^{n \times m} $. Distributions in $\R^n$ are denoted by $\rho$ and probability density functions (PDFs) by $p$. Given a random variable $x \sim \rho_x$ and a transformation $y = F(x)$, the pushforward of $x$ is denoted by $\rho_y = F_{\#} \rho_x$. 

\section{Preliminaries}

\subsection{Normalizing flows} \label{sub1A}

Let $\rho_1, \rho_2$ be two distributions with probability density functions $p_1(x), p_2(x)$ respectively. Their KL divergence is given by
\begin{subequations} \label{KLdiv}
    \begin{align}
    \hspace*{-3mm}
    \KLdiv(\rho_1  \|  \rho_2) & = \int_{-\infty}^{\infty} p_1(x) \log \frac{ p_1(x)}{ p_2(x)} \mathrm{d}x \label{KLdiv_analytic}\\
    & = \E_{x \sim \rho_1} \left[ \log \frac{p_1(x)}{p_2(x)} \right ] \\
    & = \E_{x \sim \rho_1} \left[ \log p_1(x) \right] - \E_{x \sim \rho_1} \left[ \log p_2(x) \right ] \label{KLdiv_simple}.
    \end{align}
\end{subequations}
Equation \eqref{KLdiv} suggests that calculating the KL divergence requires the analytic expressions of the PDFs of the two functions $\rho_1$ and $\rho_2$.
In the context of Problem \eqref{reference_problem_soft}, $\rho_1$ would be a target distribution, which we know explicitly, while $\rho_2$ would correspond to the distribution of the state at some time step of interest. 
The calculation of $\rho_2$ would therefore require propagating the initial state distribution through the nonlinear dynamics, which is generally intractable.
To overcome this issue, one can use the change of variables formula connecting the PDFs of two random variables that are linked through a diffeomorphic (invertible and differentiable) transformation. 
This is summarized in the following lemma:
\begin{lemma}~\cite{peyre2017computational} (Change of Variables)\label{change_of_vars} 
Let x be an $n$-dimensional random variable with known PDF, denoted $p_x(x)$, and let $F: \R^n \rightarrow \R^n$ be a diffeomorphism. 
The PDF of $y = F(x)$, denoted $p_y(y)$, can be calculated using the formula
\begin{equation} \label{f_change_of_vars}
    p_y(y) = p_x \left(F^{-1}(y) \right) \det\; \nabla F^{-1}(y).
\end{equation}
% Alternatively, the PDF of $x$ can be expressed in terms of the PDF of $y$ as 
% \begin{equation}
%     p_x(x) = p_y \left(F(x) \right) \det \; \nabla F(x)
% \end{equation}
% where we have used the properties $\nabla F^{-1}( F(x)) = (\nabla F(x))^{-1}$ and $\det \; (\nabla F(x))^{-1} = 1/\det \; \nabla F(x) $.
\end{lemma}
Transformations that correspond to flow maps of continuous-time systems are invertible, as far as unique solutions exist for the differential equation describing their dynamics. 
In CNF problems without prior dynamics, this condition is satisfied because neural networks with finite weights and Lipschitz nonlinearities result in Lipshitz ODE dynamics, and the uniqueness of solutions is guaranteed through Picard's existence theorem  \cite{chen2018neural}.
In the discrete-time case, however, the invertibility of the network needs to be addressed explicitly. 
In this paper, we make use of the following lemma:  
\begin{lemma} ~\cite{behrmann2019invertible} (Flow Invertibility) \label{Invertibility}
Consider the discrete-time nonlinear system described by 
\begin{equation}
    x_{k+1} = x_k + f_k(x_k), \quad k = 0, 1, \dots, N-1,
\end{equation}
and the transformation $x_N = F(x_0)$. The transformation $F$ is invertible if,
for all $k = 0, 1, \dots, N-1$, the mappings $f_k$ are contractive.
\end{lemma}
For a different technique that requires modeling the individual state transition functions as gradients of a convex potential, we refer the reader to \cite{huang2020convex}. 
Finally, one result that facilitates the derivation of the proposed algorithm is the equivalence of the KL-divergence before and after a diffeomorphic transformation.
\begin{lemma} ~\cite{papamakarios2021normalizing} (Forward and Backward KL divergence) \label{FB_KL}
    Let $F:\R^n\rightarrow \R^n$ be a diffeomorphism and let $\rho_1, \rho_2$ be two distributions. Then 
    \begin{equation}
        \KLdiv (\rho_1 \| \rho_2) = \KLdiv (F_{\#} \rho_1 \| F_{\#} \rho_2 ).
    \end{equation}
\end{lemma} 

\subsection{Exact Steering Using Semidefinite Programming}  \label{sub:SDP}

For the special case of a system with linear dynamics of the form
\begin{equation} \label{lin_dyn}
    x_{k+1} = A x_k + B u_k,   
\end{equation}
and for Gaussian initial and terminal state distributions, the optimal policy for Problem \eqref{reference_problem} is parametrized by an affine feedback controller of the form \cite{liu2022optimal}
\begin{equation} \label{affine_policy}
    \pi_{k}(x_k) = K_k (x_k - \mu_k) + v_k,
\end{equation}
where $\mu_k = \E[x_k]$, while its solution, i.e., the calculation of $\{K_k, v_k \}$ for $k = 0, 1, \dots N-1$, can be attained efficiently through semidefinite programming.
The soft-constrained version \eqref{reference_problem_soft}
has only been studied for a Wasserstein-2 soft constraint penalty function in \cite{balci2022exact}.
The KL-divergence soft-constrained version can be solved similarly.
Although this paper focuses on the nonlinear version of the problem, the case of linear prior dynamics with Gaussian boundary conditions will be used as a benchmark, to validate the accuracy of the proposed algorithm. 
To this end, we briefly present how one can calculate the optimal solution to Problem (\ref{reference_problem_soft}) with a controller of the form \eqref{affine_policy} using semidefinite programming. 
Although we do not explicitly prove that this family of policies is globally optimal for the problem in this paper, this has been shown to be the case for the hard constrained Problem \eqref{reference_problem} 
in~\cite{liu2022optimal}.

\subsection{KL-divergence Covariance Steering} \label{KL-cov}

When the dynamics of the system are linear and a control policy of the form \eqref{affine_policy} is used, the first two moments of the state can be calculated explicitly at any time instant $k$ in the steering horizon.
The corresponding equations are
\begin{subequations} \label{moment_prop}
\begin{align}
    & \mu_{k+1} = A\mu_k + B v_k, \\
    & \Sigma_{k+1} = (A + B K_k) \Sigma_k (A + B K_k)\t,
\end{align}
\end{subequations}
where $\Sigma_k = \cov(x_k)$. The KL-divergence between the terminal distributions $\rho_N = \mathcal{N}(\mu_N, \Sigma_N)$ and $\rho_f = \mathcal{N}(\mu_f, \Sigma_f)$ can be calculated via ~\cite{duchi2007derivations}
\begin{subequations}
\begin{align*}
    \KLdiv(\rho_N \| \rho_f) = &  \half  ( \tr(\Sigma_f^{-1}\Sigma_N) + (\mu_f - \mu_N)\t \Sigma_f^{-1} (\mu_f - \mu_N)\\
    & + \log \det(\Sigma_f) - \log \det (\Sigma_N)  - n ),
\end{align*}
\end{subequations}
where $n$ is the dimension of the state vector. 
Using equations \eqref{moment_prop}, the change of variables $K_k = \Sigma_k^{-1} U_k$, and $Y_k = U_k \Sigma_k^{-1} U_k$ in 
Problem~\eqref{reference_problem_soft} yields
\begin{subequations} \label{SDP}
\begin{align}
& \min \; J = \sum_{k = 0}^{N-1} \tr (Y_k) + \|v_k\|^2 + \lambda \KLdiv(\rho_N\|\rho_f),  \\
&  U_k \Sigma_k^{-1} U_k = Y_k \label{relaxation}, \\
& \Sigma_{k+1} = A \Sigma_k A\t + B U_k + U_k\t B \t + B Y_k B\t, \\
& \mu_{k+1} = A\mu_k + B v_k, \\
& \Sigma_0 = \Sigma_i, \\ 
& \mu_0 = \mu_i, \\
& \Sigma_N = \Sigma_f, \\
& \mu_N = \mu_f.
\end{align}
\end{subequations}
Relaxing \eqref{relaxation} to the semidefinite inequality $U_k \Sigma_k^{-1} U_k \preceq Y_k$ turns Problem \eqref{SDP} into a semidefinite program. 
This relaxation has been proven to be lossless in \cite{liu2022optimal, balci2022exact}. 
Finally, we note that the term $-\log \det(\Sigma_f)$ in the KL-divergence is convex with respect to $\Sigma_f$ and can be added to the cost function using appropriate slack variables accompanied by an LMI constraint \cite{aps2020mosek}. 
\section{Maximum likelihood Distribution Steering}

This section contains the main results of the paper. 
To this end, consider Problem \eqref{reference_problem_soft} with a policy parametrized by a neural network, i.e., $u_k = \pi_{k}(x_k ; \theta)$ where $\theta$ corresponds to the trainable policy parameters.
To optimize \eqref{reference_problem_soft}, tractable expressions for the cost function \eqref{ref_soft:cost}
%
% \begin{equation*}
%     J = \sum_{k = 0}^{N-1} { \E \left[ \| u_k \|^2  + V_k(x_k)\right]} + \lambda \KLdiv (\rho_N \| \rho_f) 
% \end{equation*}
%
must be developed. 
The first step is arguably the calculation of the KL divergence. 
Calculating it directly would involve the explicit calculation of the probability density of the state at the end of the steering horizon, $p_N$, which is challenging due to the nonlinearities of the system.
Instead of calculating $p_N$ directly, let $F(x_0) = \Phi_{N-1} \circ \Phi_{N-2} \circ \dots \circ \Phi_{0} (x_0)$ denote the transformation linking the initial and terminal states under the discrete dynamic model \eqref{ref:dyn}, where 
\begin{equation}\label{state_trans}
    \Phi_k(x_k) = f_k(x_k) + B_k \pi_k(x_k),
\end{equation} 
is the closed-loop state transition function at time step $k$.
Under certain conditions on the control policy that will be specified later in the section, this transformation is diffeomorphic.
Therefore, its inverse $x_0 = F^{-1}(x_N)$ satisfies the conditions of Lemma \ref{FB_KL}.
Applying this result to the KL divergence yields 
\begin{equation*}
    \KLdiv (\rho_N \| \rho_f) = \KLdiv (\rho_i \| F^{-1}_{\#} \rho_f).
\end{equation*}
Further expanding the second term using the definition of the KL divergence, yields 
\begin{equation*}\label{ML}
    \KLdiv (\rho_i \| F^{-1}_{\#} \rho_f)) =  \E_{x \sim \rho_i} [\log p_i(x)] - \E_{x \sim \rho_i} [\log p_{0}(x)],
\end{equation*}
where $p_i(x)$ is the PDF of $\rho_i$ and $p_0(x)$ is the PDF of the distribution $F^{-1}_{\#} \rho_f$, that is, the density of a random variable sampled from $\rho_f$ and pushed through the inverse transformation $F^{-1}$.
Notice that the term $\E_{x \sim \rho_i} [p_i(x)]$ does not depend on the control policy parameters, and can therefore be omitted from the cost function of \eqref{reference_problem_soft}. 

The calculation of $\log p_{0}(x)$ can be facilitated through Lemma~\ref{change_of_vars}. 
Specifically, one can link the density $p_{0}$ with the density of $p_f$ using 
\begin{equation*}
    \log p_{0}(x) = \log p_{f} (F(x)) + \log \det \nabla F(x).
\end{equation*}
The second term can also be efficiently calculated using the chain rule as follows 
\begin{equation*}
    \log \det \nabla F(x) = \sum_{k=0}^{N-1} \log \det \nabla \Phi_{k} (x_k).
\end{equation*}
In our implementation, the Jacobian of the state transition functions $\nabla \Phi_k(x_k)$ were calculated using automatic differentiation. 
% The expectations of $\E_{x \sim \rho_i} [\log p_{0}(x)]$ 
% and $ \E_{x \sim \rho_k} [ \| u_k \|^2 ]$ were calculated empirically, using a batch of samples. 

Finally, we discuss conditions for $\pi_k(x_k; \theta)$ that preserve the invertibility of the discrete-time dynamics.
Without loss of generality, we assume that the state transition functions \eqref{state_trans} are of the form 
\begin{equation} \label{affine_dyn}
    x_{k+1} = x_k + \phi_k(x_k) + B_k \pi_k(x_k ; \theta).
\end{equation} 
\begin{proposition} \label{thm:inv}
    Let the system dynamics be described by \eqref{affine_dyn}, and let $L_{\phi_k}, L_{\pi_k}$ be the Lipschitz constants of $\phi_k, \pi_k$, respectively, and let $\sigma_{B_k}$ be the spectral norm of the matrix $B_k$. 
    Then, if $L_{\pi_k} < (1 -L_{\phi_k})/\sigma_{B_k}$, the state transition function defined in \eqref{affine_dyn} is a diffeomorphism.  
\end{proposition}
\begin{proof}
    Based on Lemma \ref{Invertibility}, it suffices to show that $\phi(x_k) + B_k \pi_k(x_k)$ is a contraction.
    To upper bound its Lipschitz constant note that $\| \nabla \left( \phi_k + B_k \pi_k(x_k) \right) \|_2 \leq \| \nabla \phi_k \|_2 + \| B_k \nabla \pi_k \|_2 \leq L_{\phi_k} + \sigma_{B_k} L_{\pi_k}$, 
    due to the subadditivity and submultiplicativity of the spectral norm \cite{bernstein2009matrix}. 
    % \begin{subequations}
    %     \begin{align*}
    %         & \| \nabla \left( \phi_k + B_k \pi_k(x_k) \right) \|_2 \leq \\
    %         & \| \nabla \phi_k \|_2 + \| B_k \nabla \pi_k \|_2 \leq \\ 
    %         & L_{\phi_k} + \sigma_{B_k} L_{\pi_k}
    %     \end{align*}
    % \end{subequations}
    Constraining this upper bound yields the desired result. 
\end{proof}
\begin{remark}
    In the case where $f_k, B_k$ in \eqref{ref:dyn} are discretized versions of the continuous dynamics $\dot{x}_t = f_t(x_t) + B_t u_t$, then the first order approximation of the terms in \eqref{affine_dyn} are $\phi_k(x_k) = \Delta T f_t(x_k)$ and $B_k = \Delta T B_t$, where $\Delta T$ is the discritzation step size. 
    Therefore, for Lipschitz continuous-time dynamics, $L_{\phi_k}$ and $\sigma_{B_k}$ can be made sufficiently small by reducing the discretization step $\Delta T$.
\end{remark}

Note that training Neural Networks with bounded Lipschitz constants can be achieved using spectral normalization \cite{miyato2018spectral}. 
In this work, we use $\pi_k = \alpha L_{\pi_k} \hat{\pi}_k$ where $L_{\pi_k} = (1 -L_{\phi_k})/\sigma_{B_k}, \; \alpha \in (0 , 1)$ and $\hat{\pi}_k$ is a Multilayer Perceptron (MLP) with spectral normalization in all of its weights, having therefore a Lipschitz constant of 1.

% The details of the proposed approach are summarized in Algorithm~\ref{alg:cap}.
After calculating the loss function, optimization is carried out using standard gradient-based optimizers. 
For our implementation, we used the AdamW \cite{loshchilov2017decoupled} scheme, implemented in pytorch \cite{paszke2019pytorch}.
\section{Numerical examples}

\begin{figure*}
    \centering
    \begin{subfigure}[t]{0.33\textwidth}
        \centering
        \includegraphics[height=1.5in]{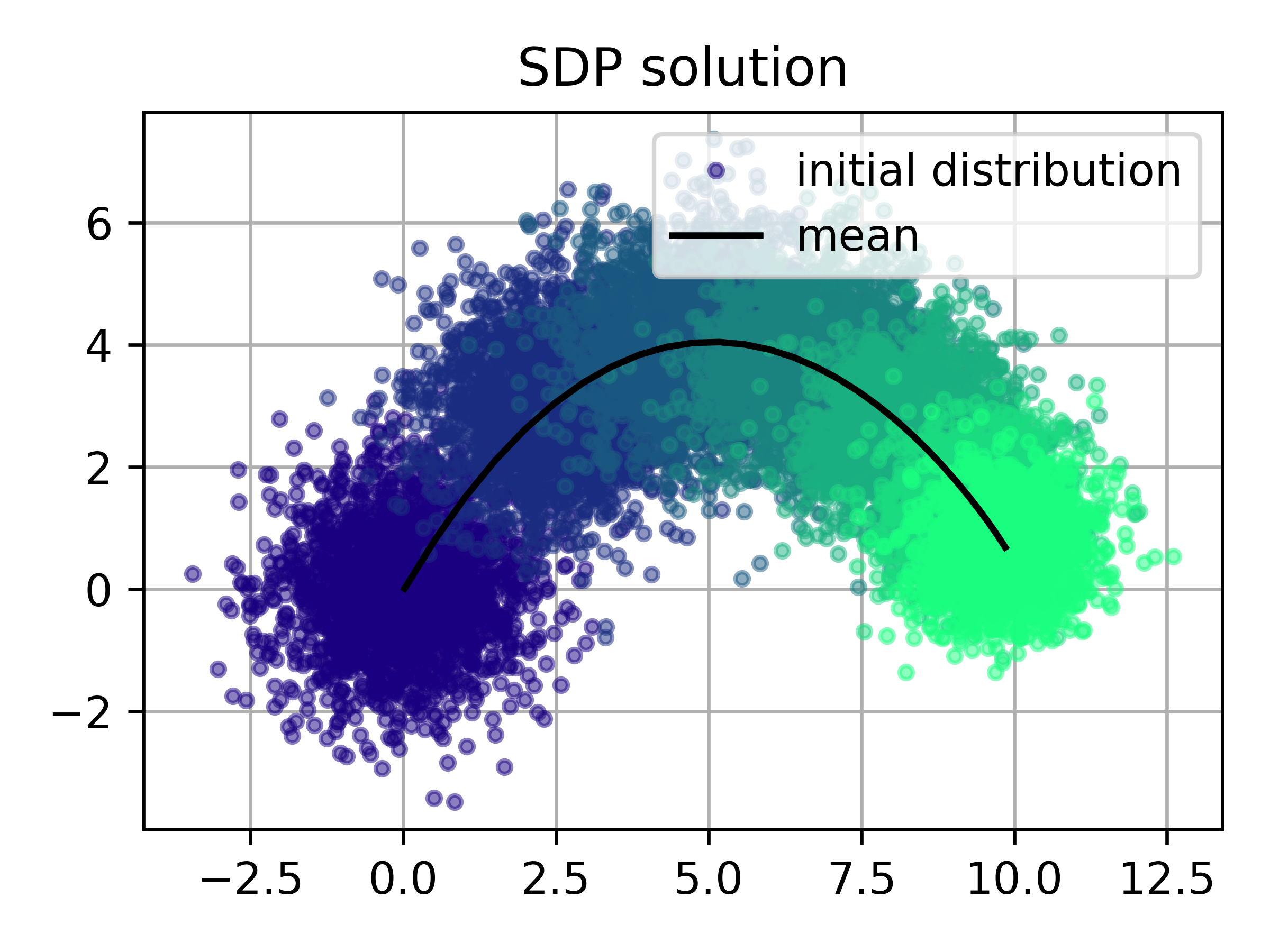}
        \caption{}
    \end{subfigure}%
    \begin{subfigure}[t]{0.33\textwidth}
        \centering
        \includegraphics[height=1.5in]{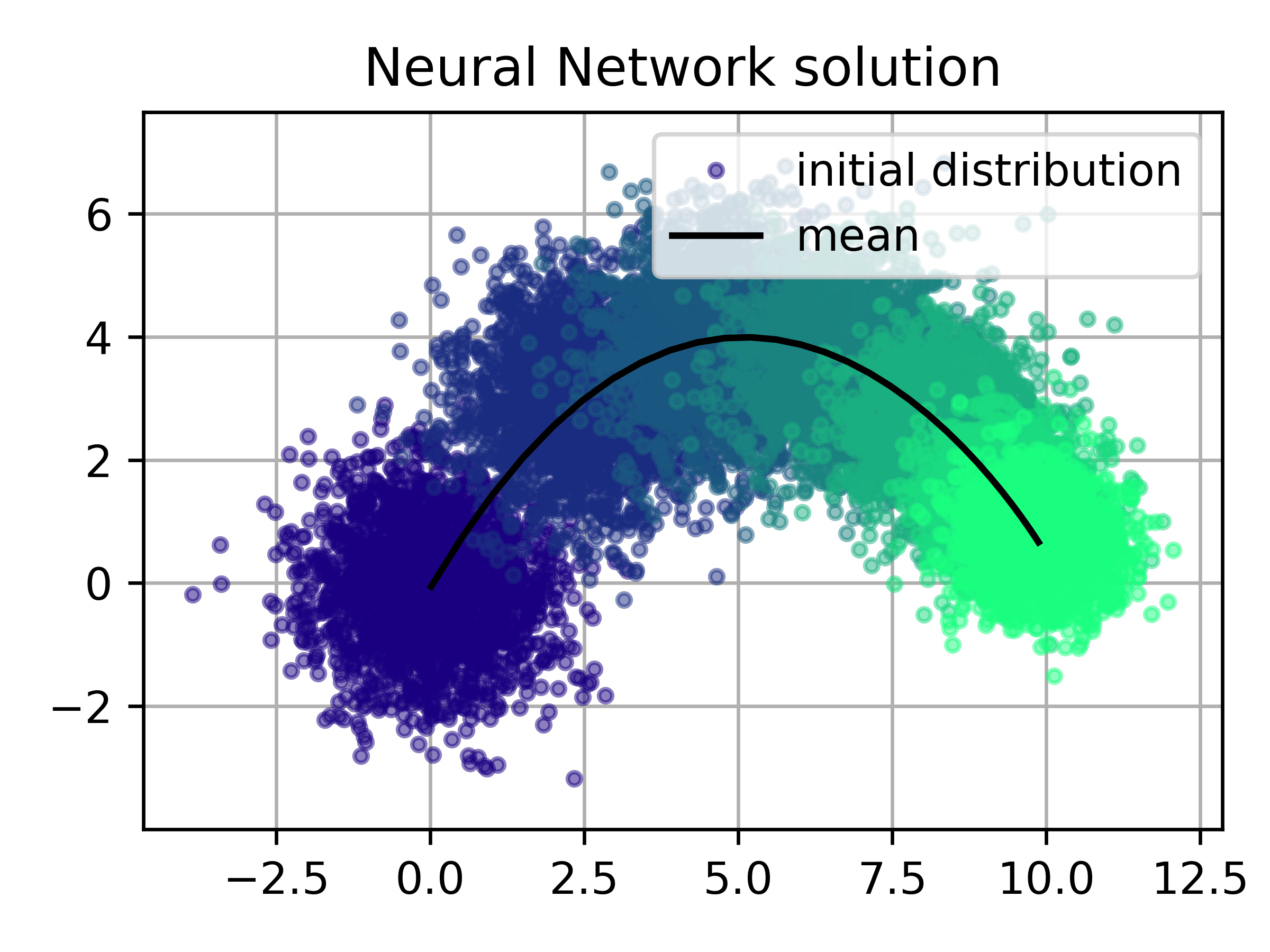}
        \caption{}
    \end{subfigure}
    \begin{subfigure}[t]{0.33\textwidth}
        \centering
        \includegraphics[height=1.5in]{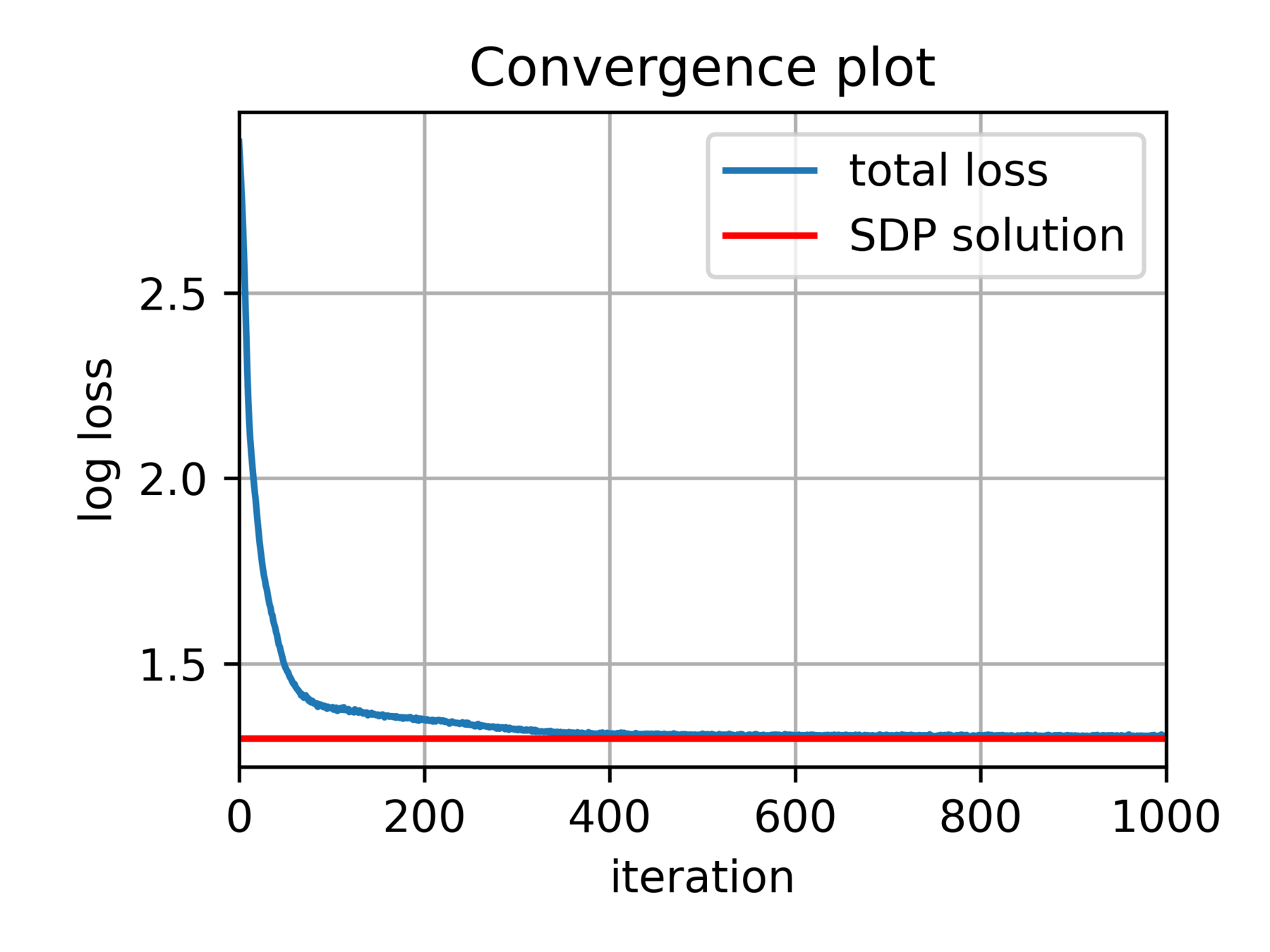}
        \caption{}
    \end{subfigure}
    \caption{(a) Exact SDP solution (b) Neural Network solution (c) Convergence plot along with optimal cost calculated using the SDP method. For figures (a), (b) the axes correspond to the first two states of the 2D double integrator.}
    \label{fig:comparison_cc}
\end{figure*}
In the first numerical example, we study the problem of driving a double integrator system from an initial to a terminal Gaussian distribution. 
Since an exact solution can be obtained for this problem using the results from Section~\ref{KL-cov}, we use this example as a benchmark.
To this end, consider the discrete-time deterministic dynamics of the form \eqref{lin_dyn} with
\begin{equation*}
    A = \begin{bmatrix} \ I_2 & \Delta T I_2 \\ 0_2 & I_2 \end{bmatrix}, \quad B = \begin{bmatrix} 0_2 \\ \Delta T I_2  \end{bmatrix}, \quad \Delta T = 0.1,
\end{equation*}
a horizon of $N=30$ time steps, $V_k(x_k) \equiv 0$ and $\lambda = 60$, which is equivalent with normalizing $\| u_k \|^2$ with $ 1/(N m)$ where $m$ is the number of input channels. We use this value for $\lambda$ for all the subsequent examples. 
The boundary distributions are $x_0 \sim \mathcal(\mu_i, \Sigma_f)$ and $x_N \sim \mathcal{N}(\mu_f, \Sigma_f)$ with parameters $\mu_i = [0,\; 0, \; 5,  \; 8], \quad \Sigma_i = \blkdiag(1, 1, 0.2, 0.2), \mu_f = [0,\; 0, \; 10,  \; 0], \quad \Sigma_f = 0.4 I_4$.
For this example, the Lipschitz constants of the prior dynamics are $L_\phi = \sigma_B = \Delta T$. 
To this end, we set $L_{\pi} = 9, \alpha = 0.9$ and $\pi_k = \alpha L_{\pi} \hat{\pi}_k$.
Each policy $\hat{\pi}_k(\cdot)$ is modeled using a fully connected MLP with spectral normalization, and five layers with $\{4, 64, 64, 64, 64, 2\}$ neurons per layer.
The convergence plot, along with the optimal solution calculated using the SDP technique described in Section~\ref{KL-cov} can be viewed in Figure \ref{fig:comparison_cc}. 
%
% It is worth noting that the proposed approach achieves the optimal policy as calculated from the SDP solution.
%
\begin{figure*}
    \centering
    \includegraphics[width=1 \textwidth]{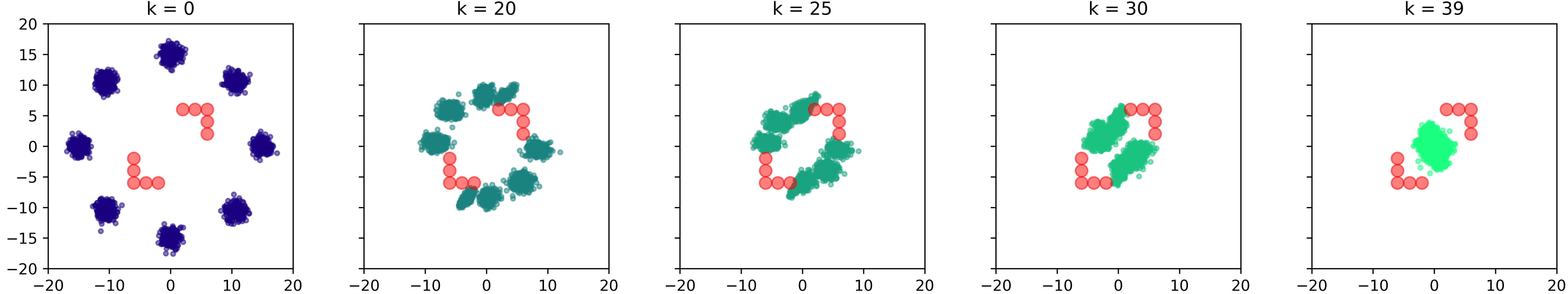}
    \caption{GMM to Gaussian with given mean and covariance, double integrator prior dynamics and obstacles.}
    \label{fig:GMM2G_4D}
\end{figure*}
\begin{figure*}
    \centering
    \includegraphics[width=1 \textwidth]{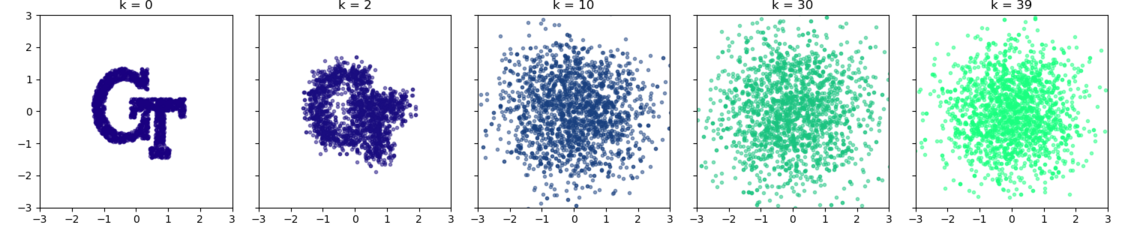}
    \caption{GT to Gaussian distribution steering with given mean and covariance with double integrator prior dynamics.}
    \label{fig:GT2G}
\end{figure*}
\begin{figure*}
    \centering
    \includegraphics[width=1 \textwidth]{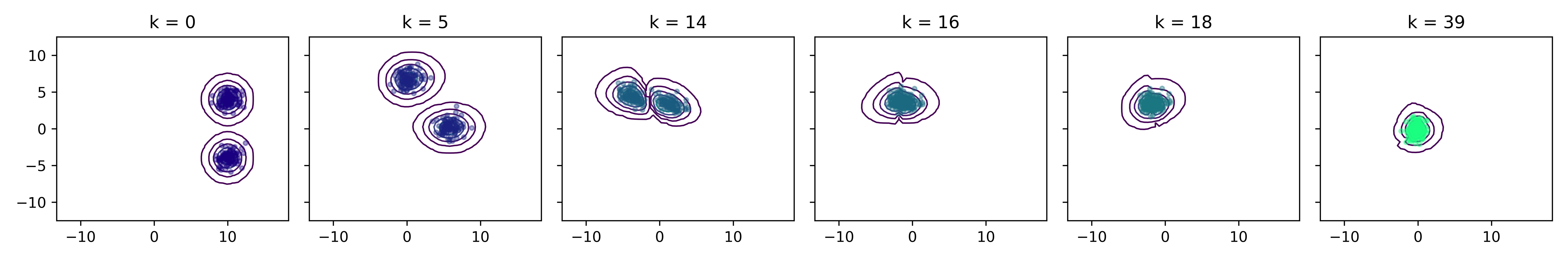}
    \caption{GMM to Gaussian distribution steering with nonlinear prior dynamics. The axes correspond to the system states.}
    \label{fig:G2GMM_nonlin}
\end{figure*}

In the second numerical example, we use double integrator dynamics, a horizon of $N=40$, but this time, we opt to steer from a Gaussian Mixture Model with 8 modes to the Normal distribution in the presence of obstacles.
The obstacles are modeled using appropriate potential fields with Gaussian kernels \cite{onken2022neural} of the form
\begin{equation}
    V_k(x) = \lambda_{\textrm{obs}} \exp \left( -\frac{(x - x_0)^2}{r_{\textrm{obs}}^2} \right).
\end{equation}
The policy at each time step has the same structure as in the first example but with 128 neurons in the hidden layers. 
The results are illustrated in Figure~\ref{fig:GMM2G_4D}.

Another, more complicated example, which capitalizes on the fact that only a batch of samples is required for the initial distribution rather than its PDF, is depicted in Figure~\ref{fig:GT2G}, where the initial distribution is arbitrary and the terminal is the normal distribution. 
The prior dynamics and policy parametrization are identical to Example~2.
We note that the inverse problem, i.e., steering from a Gaussian distribution to an arbitrary distribution for which only samples are available, would be significantly harder and cannot currently be solved with the proposed approach, since explicit information about the PDF of the terminal distribution is required for the maximum likelihood training.
We leave the investigation of this case as part of future work.

Finally, we test the proposed method in the 2D nonlinear model 
\begin{subequations}
    \begin{align}
        & x_{k+1} = x_k + 0.1 \sqrt{1 + y_k^2} + u_k, \\
        & y_{k+1} = y_k + 0.1 x_k,
    \end{align}
\end{subequations}
for $N=40$.
By bringing these equations in the form of \eqref{affine_dyn} with  $L_B = \sigma_B = 0.1$, 
we may use the same policy parametrization as in the first example. 
The results are illustrated in Figure~\ref{fig:G2GMM_nonlin}.
Since this is a two-dimensional example, we can overlay the samples with the contour lines of the PDF at each time step to validate the precision of the solution.
% We observe that even though this is a soft-constrained problem with respect to the target terminal distribution condition, the last is achieved reasonably well in the final time step.

Table~\ref{tab:quant} demonstrates a quantitative evaluation of the algorithm's performance. 
The first column corresponds to the experiment number. 
To measure the distance between the distribution of the state at the final time step of the steering horizon and the target distribution, we calculate the 2-Wasserstein distance using discrete Optimal Transport \cite{peyre2017computational} and report its value in the second column.
We use the 2-Wasserstein distance because it can be computed exactly on empirical distributions that are available only through samples and accurately reflects the actual distance between the continuous distributions given enough samples.
In the third column, we report the minimum value of the log-determinant of the Jacobian of the optimal map linking the initial and final state in order to validate the invertibility of the computed map and finally provide the total training time in minutes in the last column.
Training was performed on an Nvidia RTX-3070 GPU. 

% \panos{Wasserstein metric has not been introduced anywhere in the paper. Should we do it? It is not clear what we need to compare with the Wasserstein metric. Maybe elaborate a bit more?}

\begin{table}[h]
    \caption{Quantitative analysis of the proposed approach.}
    \label{tab:quant}
    \centering
    \begin{tabular}{|c|c|c|c|}
        \hline
         Exp. \# & $\mathbb{W}(\rho_N \| \rho_f)$ & min $ | \log \det \nabla F | $ & Training time [m] \\
        \hline
         1 &  1.12 & 0.32 & 4.1 \\
        \hline
         2 &  1.56 & 3.19 & 20.1\\
        \hline 
         3 & 0.64 & 0.42 & 13.53\\
        \hline 
         4 & 0.18  & 0.17 & 10.5 \\
        \hline 
    \end{tabular}
\end{table}
\section{Conclusions}
This paper presents a method for solving the distribution steering problem for discrete-time nonlinear systems by formulating it as a regularized maximum likelihood optimization problem.
The control policies are parametrized using neural networks with appropriate Lipschitz constraints to ensure the invertibility of the discrete-time dynamics.
A general cost function is considered, allowing state-dependent terms to model obstacles in the state space using potential fields. 
In parallel, a KL-divergence soft constraint version of the Covariance Steering problem is developed as a benchmark to compare with the proposed nonlinear maximum likelihood methods. 
Finally, four comprehensive numerical examples are presented and analyzed with respect to how closely they achieve the target distribution, as well as in terms of run time.
For the linear dynamics with Gaussian boundary distributions, the solution is also compared against the globally optimal solution calculated as the solution of a semidefinite program. 
\section*{ACKNOWLEDGMENT}
The authors would like to sincerely thank Dr. Ali Reza Pedram for his comments in an initial version of the manuscript and for multiple fruitful discussions. 
Support for this work has been provided by 
ONR award N00014-18-1-2828 and NASA ULI award \#80NSSC20M0163.
This article solely reflects the opinions and conclusions of its authors and not of any NASA entity.

% \panos{Fix all missing info in several refs. Also, fix capitalization of proper names}

\bibliographystyle{ieeetr}
\bibliography{refs}

\end{document}